\documentclass[runningheads]{llncs}

\usepackage{microtype}
\usepackage{amsmath,enumitem}
\usepackage{amssymb}
\usepackage{graphicx}
\usepackage{wrapfig}
\usepackage[hidelinks]{hyperref}
\usepackage{color}

\let\emptyset\varnothing

\newcommand{\etal}{et al.}

\newcommand{\cloneclaim}[2]{\medskip\noindent\textbf{#1.} \emph{#2}\smallskip}

\graphicspath{{figures/}}

\author{
Arthur van Goethem\inst{1}
\and Irina Kostitsyna\inst{1}
\and Marc van Kreveld\inst{3}
\and Wouter Meulemans\inst{1}
\and Max Sondag\inst{1}
\and Jules Wulms\inst{1}
}
\authorrunning{A. van Goethem et al.}

\institute{
Dept. of Mathematics and Computer Science, TU Eindhoven\\ \email{\{ a.i.v.goethem | i.kostitsyna | w.meulemans\linebreak | m.f.m.sondag | j.j.h.m.wulms \}@tue.nl}
\and Dept. of Information and Computing Sciences, Utrecht University \email{m.j.vankreveld@uu.nl}
}

\title{
The Painter's Problem: \newline
covering a grid with colored connected polygons%
\thanks{
This work was initiated at the 2nd Workshop on Applied Geometric Algorithms (AGA 2017) supported by the Netherlands Organisation for Scientific Research (NWO), 639.023.208.
AvG is supported by NWO 612.001.102; IK by FRS-FNRS; MvK by NWO 612.001.651; WM and JW by NLeSC 027.015.G02; MS by NWO 639.023.208.
}
}
\titlerunning{The Painter's Problem: covering a grid with colored connected polygons}

\begin{document}

\maketitle

\begin{abstract}
Motivated by a new way of visualizing hypergraphs, we study the following problem.
Consider a rectangular grid and a set of colors $\chi$.
Each cell $s$ in the grid is assigned a subset of colors $\chi_s \subseteq \chi$ and should be partitioned such that for each color $c\in \chi_s$ at least one piece in the cell is identified with $c$. Cells assigned the empty color set remain white.
We focus on the case where $\chi = \{\text{red},\text{blue}\}$.
Is it possible to partition each cell in the grid such that the unions of the resulting red and blue pieces form two connected polygons?
We analyze the combinatorial properties and derive a necessary and sufficient condition for such a \emph{painting}.
We show that if a painting exists, there exists a painting with bounded complexity per cell.
This painting has at most five colored pieces per cell if the grid contains white cells, and at most two colored pieces per cell if it does not.
\end{abstract}

\section{Introduction}
\label{sec:introduction}
Hypergraphs are a powerful structure to represent unordered set systems.
In general, there are a number of elements (vertices of the hypergraph) and a number of different subsets over these elements (the hyperedges of the graph).
The purpose of visualizing hypergraphs is to clarify the various set relations between the hyperedges.
There are, roughly speaking, two strands of hypergraph visualizations: those where the position of the elements is fixed (e.g.~\cite{Alper2011,Collins2009,Dinkla2012,Meulemans2013}), and those where the positions can be chosen by the layout algorithm (e.g.~\cite{Riche2010,Simonetto2008,Simonetto2009}).
For a more detailed overview and in-depth classification of set visualization methods we refer to the survey by Alsallakh~\etal~\cite{Alsallakh2016}.
Though some methods aim to overcome layout complexity by replicating elements (e.g.~\cite{Alsallakh2013,Riche2010}), we focus on a visualization using a single representation for each element.

In theoretic research on drawing hypergraphs (e.g.~\cite{Buchin2011,Kaufmann2009}), the (often implicit) assumption is that the representations of two sets may cross at common vertices.
Such crossings are not deemed problematic as most visual encodings rely on the local \emph{nesting} of intersecting polygons (in line with the prototypical Venn and Euler diagrams~\cite{Baron1969} and similar visual overlays \cite{Collins2009,Dinkla2012,Meulemans2013}) to identify set memberships.
Nesting, however, gives a strong visual cue of containment and may result in misleading visual representations implying containment relationships between hyperedges.
A rendering style without nesting is one suggested for Kelp Diagrams~\cite{Dinkla2012}.
However, its cluttered appearance caused it not to feature in the later extension, KelpFusion \cite{Meulemans2013}.

One of the most well-established quality criteria of graph drawings is planarity (see e.g. \cite{Purchase2002,Purchase1996}).
When nested encodings are used, a planar drawing relates to finding a planar support~\cite{Buchin2011}: a planar (regular) graph such that the vertices of each hyperedge induce a connected subgraph in the support.
Deciding whether a planar support exists is possible for some simple support classes (see~\cite{Buchin2011} for a discussion), but is already NP-hard for 2-outerplanar support graphs~\cite{Buchin2011}.
Optimizing hypergraph supports for total graph length without planarity constraints is
NP-hard, but approximation algorithms exist \cite{Akitaya2016,Hurtado2013}.

Representations that do not require nesting are edge-based drawings~\cite{Makinen1990} or the equivalent Zykov representation~\cite{Walsh1975}, for which notions of planarity follow readily from the standard notion for regular graphs.

\begin{wrapfigure}[13]{r}{0.33\linewidth}
  \centering
  \vspace{-0.8\intextsep}
  \includegraphics{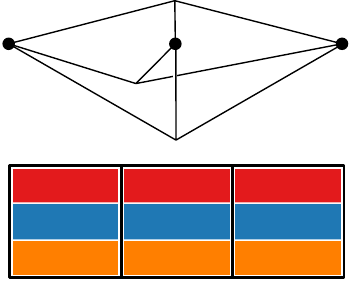}
  \caption{A hypergraph that is not Zykov-planar (top) but has a disjoint-polygons drawing (bottom).}
  \label{fig:zykovvsdisoint}
\end{wrapfigure}

Instead we suggest a visual design that uses \emph{disjoint polygons} to present hyperedges: vertices are represented as simple geometric primitives (e.g. a square or circle); hyperedges are represented as connected polygons that overlaps only and all its incident vertices; and all such polygons are pairwise disjoint.
As illustrated in Fig.~\ref{fig:zykovvsdisoint}, our disjoint-polygons encoding is stronger as it can visualize some hypergraphs that are not Zykov-planar, whereas any Zykov-planar hypergraph admits a disjoint-polygons representation.
We can use vertices to ``pass in between'' the representations of other hyperedges, though not as flexibly as is allowed for planar supports: the polygons must remain disjoint.

\paragraph{Contributions}
We investigate the properties of drawing hypergraphs using disjoint polygons.
Motivated by moving towards a set visualization in a geographic small multiples or grid map (see e.g. \cite{Meulemans2017,Wood2008}), we specifically study the variant where each element has a fixed location, being a cell in a rectangular grid.
As an initial exploration we focus on the 2-color case, where each cell is either red, blue, both (purple), or uncolored (white).
We thus aim to partition each purple cell into red and blue pieces, such that the resulting pieces of a single color form a connected polygon.
We derive a necessary and sufficient condition to efficiently recognize whether an instance is solvable.
For solvable instances, we bound the number of colored pieces within each cell by a small constant and show that these bounds are tight.
Due to space constraints, some proofs have been shortened or omitted; for full proofs, please refer to the appendix.

\section{Preliminaries}
\label{sec:prelims}

We define a \emph{$k$-colored grid} $\Gamma$ as a rectangular grid, in which each cell $s$ has a set of associated colors $\chi_s \subseteq \{1,\ldots,k\}$.
A \emph{fully} $k$-colored grid is the case where $\chi_s \neq \emptyset$ for all cells $s$.
Throughout this paper, we primarily investigate $2$-colored grids and use \emph{colored grid} to refer to the 2-colored case, unless indicated otherwise.
We refer to the two colors as ($r$)\emph{ed} and ($b$)\emph{lue}; cells for which $\chi_s = \{r,b\}$ are called ($p$)\emph{urple}.
Cells with no associated colors are \emph{white}.

A region is a maximal set of cells that have the same color assignment ($r$, $b$, or $p$) and where every cell $s$ in the region is connected via adjacent cells to every other cell $s'$ in the region. Cells are considered adjacent if they are horizontally or vertically adjacent.

A \emph{panel} $\pi_s$ for cell $s$ (with $\chi_s \neq \emptyset$) maps each color $c \in \chi_s$ to a (possibly disconnected) area $\pi_s(c)$ such that these partition the cell: that is, $\bigcup_{c \in \chi_s} \pi_s(c) = s$ and $\pi_s(c_1) \cap \pi_s(c_2) = \emptyset$ for colors $c_1 \neq c_2$.
A \emph{painting} $\Pi$ of a $k$-colored grid consists of panels $\pi_s$ for each cell $s$ with $\pi_s(c)\not=\emptyset$ for each $c \in \chi_s$ and $\pi_s(c) = \emptyset$ otherwise.
We call a painting \emph{connected} if each color forms a connected polygon: that is, $\bigcup_{s\in \Gamma} \pi_s(c)$ is a connected polygon for each color $c \in \{1,\ldots,k\}$.
For this definition, two cells sharing only a corner are \emph{not} considered connected.
Our primary interest is in connected paintings: in the remainder, we use painting to indicate a connected painting.

\section{Characterizing Colored Grids with a Painting}
\label{sec:testing}

In this section we show how to test whether a 2-colored grid admits a painting and how to find a painting if one exists.
As all completely red, blue, and white panels are fixed, finding a painting reduces to finding partitions of purple cells that ensure that the resulting red and blue polygon are connected.
We show that this connectivity is of key importance: if we can find suitable connections though the purple regions, then we can also create a partition that results in a valid panel for each cell in the purple regions.

We capture the connectivity options for the red and blue polygon using two embedded graphs, $G_r$ and $G_b$.
We construct these graphs in three steps:
\begin{enumerate}
\item Connect red (blue) regions that are adjacent along a purple region's boundary.
\item Remove holes from the purple regions by inserting connections (Section~\ref{ssec:annuli}).
\item Construct $G_r$ and $G_b$ using a gadget for purple regions (Sections~\ref{ssec:tograph} and~\ref{ssec:spiderwebs}).
\end{enumerate}
For the first step, observe that consecutive (not necessarily distinct) adjacent regions of the same color can always be safely connected via the purple region's boundary without restricting the connectivity options for the other color (see Fig.~\ref{fig:basecase-leader}).
After the first two steps, we represent the remaining red (blue) regions as vertices in $G_r$ ($G_b$).
Edges in $G_r$ and $G_b$ represent connection options through purple regions; intersections indicate a choice to connect either blue or red regions through (part of) a purple region.
The gadget for purple regions with many adjacent red and blue regions also requires some additional vertices in these graphs.
We prove that these graphs admit a simple characterization of 2-colored grids that admit a painting, as captured in the theorem below.

\begin{theorem}\label{the:dualPaintingFinal}
A 2-colored grid $\Gamma$ admits a painting if and only if the corresponding graphs $G_r$ and $G_b$ are each other's exact duals: there is exactly one blue vertex in every red face and there is exactly one red vertex in every blue face.
\end{theorem}

For explanatory reasons we start with the simplest case: purple regions with at most four neighbors and without holes (Section~\ref{ssec:tograph}).
Subsequently, we alleviate the assumption on the number of neighbors (Section~\ref{ssec:spiderwebs}) and permit holes in the purple regions, by showing how to perform Step 2 (Section~\ref{ssec:annuli}).

\subsection{Simple Purple Regions}
\label{ssec:tograph}

We assume that Step 1 has been performed and a purple region has no holes and at most four adjacent regions.
The adjacent red and blue regions of a purple region $P$ form an ordered cyclic list as they appear along the boundary of $P$ and alternate in color (due to Step 1).
Let $\kappa(P)$ denote the length of this list for $P$.
$\kappa$ is even due to color alternation, and by assumption here $\kappa(P) \leq 4$.
There can be duplicates in this list as the same red or blue region can touch $P$ multiple times.

\begin{figure}[t]
\begin{minipage}[t]{.37\textwidth}
	\centering
	\includegraphics[page=1]{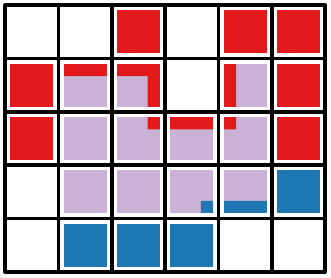}
	\caption{Safe connections between adjacent same-color regions.}
    \label{fig:basecase-leader}
\end{minipage}
\hfill
\begin{minipage}[t]{.59\textwidth}
	\centering
	\includegraphics[page=2]{base-case}
	\caption{2-colored grid with 4 regions around each purple region and corresponding graphs $G_r$ and $G_b$.}
    \label{fig:basecase-graphex}
\end{minipage}
\end{figure}

Every purple region with $\kappa(P)=2$ can be painted by creating a spanning tree on the centers of the panels of $P$ in one color and connecting it to the corresponding region. The rest of the panels is colored in the other color.
We assume these are handled; what remains is to deal with the regions with $\kappa(P) = 4$.

For a purple region $P$ with $\kappa(P) = 4$, we create a red edge in $G_r$ and a blue edge in $G_b$ that intersect: the red edge connects the red vertices corresponding to the adjacent red regions; the blue edge connects the corresponding blue vertices. There may be multiple edges between two vertices (see Fig.~\ref{fig:basecase-graphex}). If the same red or blue region touches the purple region twice, the edge is a self-loop.
Every red or blue edge intersects exactly one blue or red edge respectively, and $G_r$ and $G_b$ are plane by construction.
Using the following lemma we prove the exact characterization of graphs $G_r$ and $G_b$ of a 2-colored grid $\Gamma$ that admits a painting.

\begin{lemma}[\cite{Eppstein1992,Tutte1984}]\label{lem:cotree}
Let $G$ be a plane graph, $G^*$ its dual and $T$ a spanning tree of $G$.
Then $T^* = \{ e^* \; | \; e \not\in T \}$ is a spanning tree of $G^*$.
\end{lemma}

\begin{lemma}\label{the:dualPainting}
A 2-colored grid $\Gamma$ in which each purple region $P$ has no holes and $\kappa(P)\leq 4$, admits a painting if and only if the corresponding graphs $G_r$ and $G_b$ are each other's exact duals.
\end{lemma}

\begin{proof}[sketch]
We prove that if $\Gamma$ admits a painting then graphs $G_r$ and $G_b$ are each other's duals using a counting argument. We count the number of edges needed to connect all red and blue regions, and use Euler's formula to show the number of red faces must be equal to the number of blue vertices, and vice versa.
The other direction follows from Lemma~\ref{lem:cotree}. Having two dual spanning trees (e.g., Fig.~\ref{fig:basecase-graphex}), simply draw the two spanning trees and for any cell not yet having a blue (red) piece add a crossing-free connection to the blue (red) polygon.
\hfill$\qed$
\end{proof}

\subsection{Spiderweb Gadgets}
\label{ssec:spiderwebs}

Let us now extend the result in the previous section, by showing how to include purple regions with more than four adjacent regions.
For every purple region $P$ with $\kappa(P) > 4$ we construct a \emph{spiderweb} gadget and insert it into the graphs $G_r$ and $G_b$, such that an argument similar to Lemma~\ref{the:dualPainting} can be applied.

A spiderweb gadget $W$ of $P$ with $\kappa(P)/2 = k$ red and $k$ blue alternating adjacent regions consists of $\lfloor k/2\rfloor + 1$ levels, labeled $0$ (outermost) to $\lfloor k/2\rfloor$ (innermost), see Fig.~\ref{fig:gadgetindices}.
Each level, except $0$ and $\lfloor k/2\rfloor$, is a cycle of $k$ vertices.
The level $0$ has $k$ (blue) vertices without any edges between them, and the innermost level $\lfloor k/2\rfloor$ consists of only a single vertex.
The vertices of even levels are blue and labeled with even numbers from $0$ to $2k-2$ clockwise.
The vertices of odd levels are red and labeled with odd numbers $1$ to $2k-1$ clockwise.

Each vertex of level $\ell $ with $2 \leq \ell < \lfloor k/2 \rfloor$ is connected to the vertex with the same label on level $\ell -2$. The single vertex of level $\lfloor k/2\rfloor$ is connected to all the vertices of level $\lfloor k/2\rfloor-2$.
This gives us $2k$ paths starting from levels $0$ and $1$ to the two innermost levels. We call these paths \emph{spokes}, and refer to them by the label of the corresponding vertices.
We embed the two resulting connected components in such a way that they are each other's dual, by making sure that we get a proper clockwise numbering on the vertices of the two outermost levels (see Fig.~\ref{fig:gadgetindices}).
The vertices on levels $0$ and $1$ represent respectively the blue and red regions around the purple region $P$ and respect the adjacency order around $P$.

If a blue (or red) region touches $P$ multiple times, then the corresponding vertices on level $0$ (or $1$) map to the same region and are in fact one and the same vertex in $G_b$ (or $G_r$).
All edges connected to this vertex are consistent with the topology of the nested neighboring regions of $P$; they intersect the same edges as they would when they were represented by multiple vertices.

To prove that all possible connections in $P$, which can occur in a painting $\Pi$, can be replicated in a spiderweb gadget $W$, we define \emph{bridging paths}: let $u$ and $v$ be two vertices on level $0$ in $W$ that represent two blue regions that are connected by a painting $\Pi$ through $P$. Assume that the clockwise distance from $u$ to $v$ is not greater than $k$, that is, if $u$ has label $x$ then $v$ has label $(x+2i)\mod 2k$ for some $1 \leq i \leq \lfloor k/2\rfloor$. 
To connect $u$ and $v$ with a bridging path, we start from $u$, go to level $2\lfloor (i+1)/2 \rfloor$ along the spoke $x$, take a shortest path within the level $2\lfloor (i+1)/2 \rfloor$ from the vertex with label $x$ to the vertex with label $(x+2i)\mod 2k$, and move along the spoke $(x+2i)\mod 2k$ to vertex $v$. If there are two possible shortest paths, we take the clockwise path.

The same kind of path can be constructed for a pair of red vertices, but starting from level $1$, going to level $2\lfloor i/2 \rfloor+1$, and moving back to level $1$.
We now show that connecting different blue and red regions using bridging paths within the spiderweb gadgets results in blue trees and red trees, such that no pair of a blue and a red edge intersect (see Fig.~\ref{fig:gadgetpainting} for an example).

By performing a case analysis on the possible red and blue pairs of adjacent regions to be connected, we can prove that the following lemma holds.
\begin{lemma}\label{lem:nonintersecting}
Consider a painting $\Pi$ in which two blue and two red regions, adjacent to a purple region $P$, are connected through $P$. The corresponding vertices in the spiderweb gadget $W$ of $P$ can be connected by non-intersecting bridging paths.
\end{lemma}

With spiderweb gadgets and the above lemma, we now strengthen Lemma~\ref{the:dualPainting} to the following lemma, without a condition on $\kappa$, and prove it in a similar way.

\begin{figure}[t]
\begin{minipage}[t]{.42\textwidth}
	\centering
	\includegraphics[page=7]{spiderWebs}
	\caption{Spiderweb gadget for $k=6$: three blue levels with indices $0$, $2$, $4$, and two red levels with indices $1$, $3$.}
	\label{fig:gadgetindices}
\end{minipage}
\hfill
\begin{minipage}[t]{.54\textwidth}
	\centering
	\includegraphics{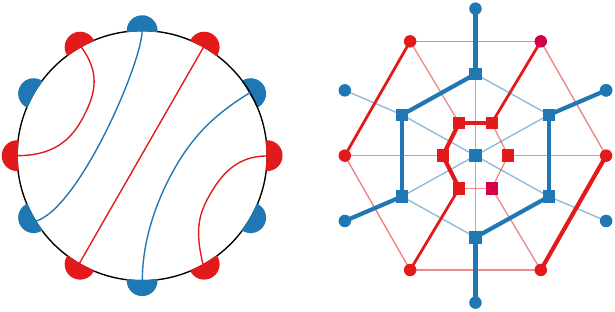}
	\caption{Topology of the connections in a purple region and the corresponding bridging paths through a spiderweb gadget.}
	\label{fig:gadgetpainting}
\end{minipage}
\end{figure}

\begin{lemma}\label{the:dualPainting2}
A 2-colored grid $\Gamma$ in which each purple region has no holes admits a painting if and only if the corresponding $G_r$ and $G_b$ are each other's exact duals.
\end{lemma}

\subsection{Purple Regions with Holes}
\label{ssec:annuli}
We may also have purple regions with holes (see Fig.~\ref{fig:annulusExample}).
We show that the number of holes can be reduced without affecting the solvability.
For simplicity of explanation we assume a region with a single hole (an annulus); regions with more holes can be reduced by considering only connections to the outer boundary.

Let $P$ be a purple annulus.
Any painting subdivides $P$ into a number of colored simple components.
Each component of color $c$ connects one or more regions of color $c$ on the boundary of $P$.
The existence of a painting is defined only by the connectivity structure of these components.
The connectivity of a component can be represented (transitively) using a set of non-intersecting simple paths (\emph{connections}) each connecting two regions on the boundary.
Let a \emph{cross-annulus connection} $\gamma_x$ be a connection between a region $x_{in}$ on the inside of the annulus and a region $x_{out}$ on the outside of the annulus.
A (connectivity) \emph{structure} is a maximal set of (pairwise non-intersecting) connections in $P$ that can be extended to a valid painting.
Let $C_S$ be the set of cross-annulus connections in a given structure $S$.
We assume the annulus is not degenerate, so red and blue regions exist both inside and outside the annulus.

\begin{lemma}
\label{lem:alternatingColors}
If a structure $S$ exists with two adjacent cross-annulus connections $\gamma_x$ and $\gamma_y$ of the same color, possibly separated by non-crossing connections, then there also exists a structure $S'$ where $C_{S'}=C_S\setminus \{\gamma_y\}$.
\end{lemma}

\begin{figure}[t]
\begin{minipage}[t]{.38\textwidth}
	\centering
	\includegraphics{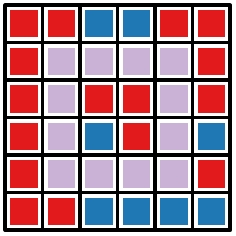}
		\caption{An annulus-type purple region with adjacent blue and red regions, both inside and outside.}
		\label{fig:annulusExample}
\end{minipage}
\hfill
\begin{minipage}[t]{.59\textwidth}
	\centering
	\includegraphics{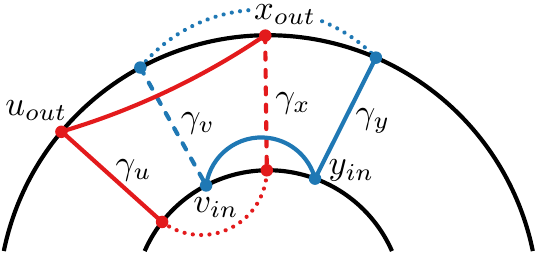}
		\caption{By adding edges $(v_{in},y_{in})$ and $(u_{out},x_{out})$ we reconnect the disconnected subpolygons formed be removing cross-annulus connections $\gamma_v$ and $\gamma_x$.}
		\label{fig:annulusRemoveAdjacent}
\end{minipage}		
\end{figure}

\begin{lemma}
\label{lem:remove2}
If there exists a structure $S$ with $|C_S|>3$ and all cross-annulus connections are alternating in color, then there also exists a structure $S'$ with $|C_S|-2$ cross-annulus connections.
\end{lemma}

\begin{proof}
Let $\gamma_u$, $\gamma_v$, $\gamma_x$, and $\gamma_y$ be four consecutive cross-annulus connections. W.l.o.g., assume $\gamma_u$ and $\gamma_x$ are red and $\gamma_v$ and $\gamma_y$ are blue.
We remove $\gamma_v$ and $\gamma_x$ from the structure separating both the red and blue into two components.
For both colors one component is still connected to the remaining cross-annulus connection $\gamma_u$, respectively $\gamma_y$.
The disconnected components cannot both be on the outside (inside) of the annulus.
If so, the connection $\gamma_u$ to $\gamma_x$ must be connected through $x_{in}$, and $\gamma_v$ to $\gamma_y$ through $v_{in}$.
However, as there is no cross-annulus connection between $\gamma_u$ and $\gamma_y$, any connection from $\gamma_u$ to $x_{in}$ separates $\gamma_y$ and $v_{in}$. Hence, both connections cannot exist at the same time.
The red and blue disconnected components are thus on different sides of the annulus and we connect them to $\gamma_u$ respectively $\gamma_y$ without mutually interfering (see Fig.~\ref{fig:annulusRemoveAdjacent}).
\hfill $\qed$
\end{proof}

\begin{corollary}
If a structure exists, then a structure also exists that has exactly one red and one blue connection across each annulus.
\label{cor:oneRedoneBlue}
\end{corollary}

\begin{lemma}\label{lem:move}
If a structure exists, then there also exists a structure with exactly one red and one blue cross-annulus connection starting from any two regions on the inner annulus and connecting to any two regions on the outer annulus.
\end{lemma}
\begin{proof}
Let an \emph{interval} be a maximal arc of the same color on the boundary.
By Corollary~\ref{cor:oneRedoneBlue} we know there exists a structure with exactly one red and one blue connection across the annulus.
Let $\gamma_x$ be the blue cross-annulus connection and $\gamma_y$ the red cross-annulus connection.
We show that each of the endpoints of the cross-annulus connection can freely be moved.
W.l.o.g., assume that $\gamma_x$ is not counter-clockwise adjacent to $\gamma_y$ on the outside of the annulus.
Let $k_{out}$, $l_{out}$, and $m_{out}$ be three intervals in clockwise order on the outer boundary of the annulus.
We say a clockwise connection through the annulus from $k_{out}$ to $m_{out}$ \emph{covers} $l_{out}$.

Let $b_{out}$ be the blue interval that is counter-clockwise adjacent to $y_{out}$ and $r_{out}$ the red interval that is counter-clockwise adjacent to $b_{out}$.
Interval $b_{out}$ may have several incoming blue connections that cover $r_{out}$ (see Fig.~\ref{fig:annulusAdjacent2}(a)).
We can rewire the blue connections inside the annulus to connect the blue intervals in sorted order around the annulus, resulting in only one blue connection $\gamma_b$ that covers $r_{out}$.
Similarly we can also rewire the red connections covering $r_{out}$, and ending at $y_{out}$, to ensure only one red connection $\gamma_r$ covers $r_{out}$.

	\begin{figure}[t]
	\centering
	\includegraphics{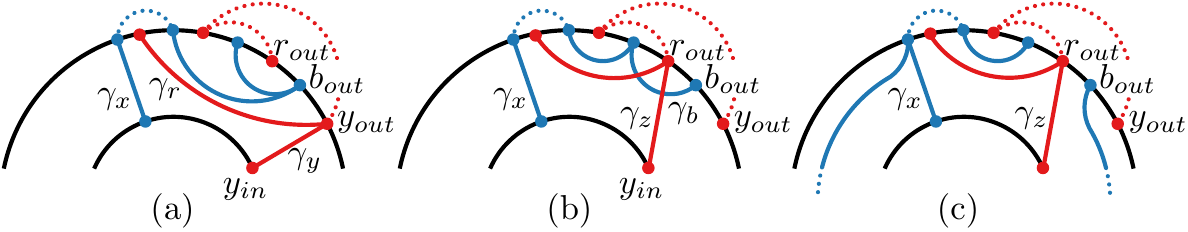}
		\caption{(a) Initial configuration with several connections covering $r_{out}$. (b) Rerouting the blue connections, introducing $\gamma_z$, and rerouting the intersecting red connection leaves only one intersecting (blue) connection. (c) As the blue disconnected component cannot be covered by the new red connection, we can always connect it back to $\gamma_x$.}
		\label{fig:annulusAdjacent2}
	\end{figure}

Remove $\gamma_y$ and insert a new red cross-annulus connection $\gamma_z = (y_{in},r_{out})$.
The connection $\gamma_z$ can only intersect $\gamma_r$ and $\gamma_b$.
Removing $\gamma_r$ results in two red components, one of which contains $\gamma_z$.
Assume w.l.o.g. that $y_{out}$ is in the same connected component as $\gamma_z$. As $\gamma_r$ intersected $\gamma_z$, the disconnected component can be connected to $\gamma_z$ while only intersecting $\gamma_b$ (see Fig.~\ref{fig:annulusAdjacent2}(b)).

Removing $\gamma_b$ results in two blue components, one of which contains $\gamma_x$.
We prove that $b_{out}$ must be part of the blue component not containing $\gamma_x$.
Assume to the contrary that $b_{out}$ is still connected to $\gamma_x$.
Interval $r_{out}$ must be connected to $y_{out}$ outside of the annulus as there was only one red cross-annulus connection and $\gamma_b$ blocked any connection through the inside of the annulus.
Similarly, interval $b_{out}$ must have been connected through the outside of the annulus, as it is separated from any other region inside the annulus by $\gamma_y$ and $\gamma_z$.
However, they cannot both be connected through the outside of the annulus, as the connection $r_{out}$ to $y_{out}$ separates $b_{out}$ and $x_{out}$ on the outside of the annulus.
Contradiction.

Therefore, we can safely reconnect the disconnected blue component through the annulus to $\gamma_x$ (see Fig.~\ref{fig:annulusAdjacent2}(c)).
Repeatedly moving the end-point of one of the cross-annulus connections allows the creation of any configuration of the two red and blue cross-annulus connections without invalidating the structure.\hfill $\qed$
\end{proof}

Lemma~\ref{lem:move} implies that we can cut the annulus open to reduce the number of holes of a purple region by one without changing the solvability of the problem.
Together with Lemma~\ref{the:dualPainting2}, this then implies Theorem~\ref{the:dualPaintingFinal}.

\section{Optimizing Panels}
\label{sec:painting}

\begin{wrapfigure}[8]{r}{0.33\linewidth}
  \centering
  \vspace{-2\baselineskip}
  \includegraphics[page=2]{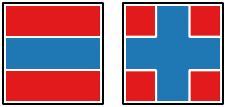}
  \caption{Panels with complexity $3$ and $5$ respectively.}
  \label{fig:panelcomplexity}
\end{wrapfigure}
As shown, not all colored grids admit a painting.
Here we investigate the design of the panels themselves, assuming
that some painting is possible. 
To this end, we define the complexity of a panel as
the number of \emph{pieces} of maximal red and blue areas in the panel, see Fig.~\ref{fig:panelcomplexity}.
The complexity of a painting is the maximal complexity of any of its panels.
A $t$-panel ($t$-painting) has complexity $t$.

Assuming some painting exists, we prove in this section that a $5$-painting exists in general and that even a $2$-painting exists if there are no white cells.

\subsection{Ensuring a 5-painting}
\label{ssec:semiregular}

We prove here that a $5$-painting can always be realized.
To this end, we show that a valid painting for a colored grid
 can be redrawn to include no more than three colored intervals along each side of all panels.

\begin{figure}[b]
	\begin{minipage}[t]{.75\textwidth}
		\centering
        \includegraphics[]{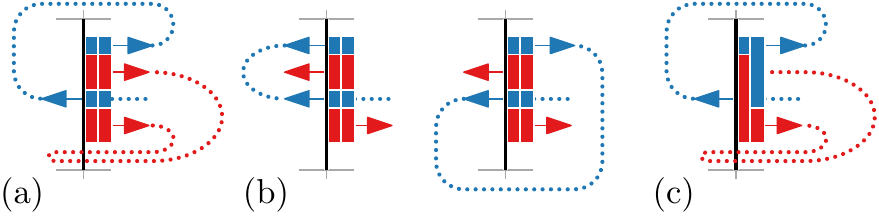}
        \caption{Reducing the number of intervals along a side of panel $\pi$, where there are at least four. (b) The two middle directions cannot be the same, as we cannot connect them with nonintersecting paths. (c) Shortcutting inside $\pi$ reduces the number of intervals while maintaining a painting. }
        \label{fig:borderColors}
	\end{minipage}
	\hfill
	\begin{minipage}[t]{.22\textwidth}
	\centering
	\includegraphics{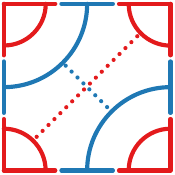}
	\caption{A panel with six pieces can always be reduced to have five, using either dotted line.}
	\label{fig:5-painting}
	\end{minipage}
	
\end{figure}

\begin{lemma}\label{lem:threeBorders}
If a 2-colored grid admits a painting, then it admits a painting where each panel $\pi$ has at most 3 intervals of alternating red and blue along each side.
\end{lemma}
\begin{proof}[sketch]
Assume that a panel $\pi$ has at least 4 intervals of alternating red and blue on the left-side of $\pi$.
As the painting is valid, both blue (/red) intervals are connected in the painting.
For each interval we identify whether the path exiting or entering $\pi$ connects to the other interval of the same color (see Fig.~\ref{fig:borderColors}(a)).
The red and blue path cannot leave or exit the border of $\pi$ in the same direction for the middle two intervals (see Fig.~\ref{fig:borderColors}(b)).
To reduce the number of intervals, we recolor the interval by shortcutting both the blue and the red piece inside $\pi$ (see Fig.~\ref{fig:borderColors}(c)).
\hfill\qed
\end{proof}

\begin{theorem}
\label{the:5painting}
If a partially 2-colored grid admits a painting, then it admits a 5-painting.
\end{theorem}

\begin{proof}
By Lemma~\ref{lem:threeBorders} there are at most three alternatingly colored intervals along each side of $\pi$.
If a red and blue interval meet in a corner, we extend one in $\pi$ around the corner to get four intervals and use Lemma~\ref{lem:threeBorders} to reduce it back to at most three.
If we have more than five pieces, a piece that has only one interval in $\pi$ can be removed while maintaining a painting.
Each remaining piece connects at least two intervals: with $k$ intervals, the number of pieces is at most $\lfloor k/2 \rfloor$.
A $6$-panel thus requires $12$ intervals: four equal-color (red) corners and a middle interval (blue) along each side.
This enforces two pieces between the blue intervals, and one in each corner.
However, we can now reduce the number of pieces to five, connecting either two blue pieces or two red corners (Fig.~\ref{fig:5-painting}).
\hfill\qed
\end{proof}

\begin{figure}[t]
\centering
\includegraphics[]{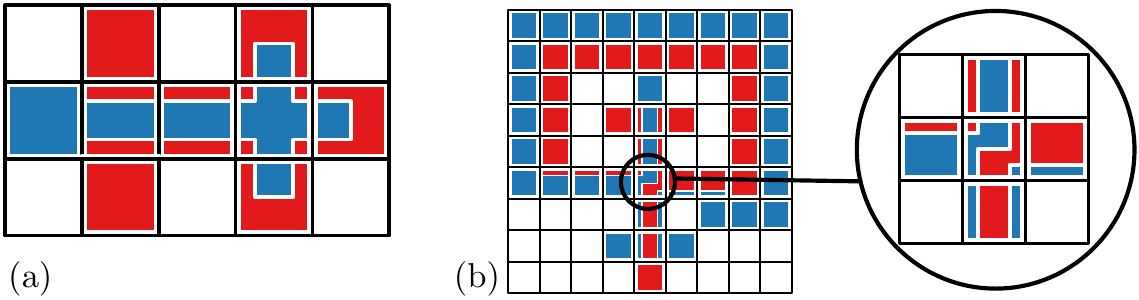}
\caption{Examples requiring complex panels. (a) A colored grid requiring a 5-panel. (b)~A colored grid requiring a 4-panel with two pieces of both colors in the same cell.}
\label{fig:complexPanels}
\end{figure}

This bound is tight as a 5-panel may be required when the grid includes white cells (Fig.~\ref{fig:complexPanels}(a)).
A 5-panel with at least two pieces of each color is never required---though such a 4-panel may be necessary (Fig.~\ref{fig:complexPanels}(b)).
The above proof implies that there is only one option to create such a 5-panel:
it has only two ways to connect the two blue pieces; both can be simplified to a 4-panel (Fig.~\ref{fig:2-3-panel-not-needed}).

\begin{figure}[t]
\centering
\includegraphics{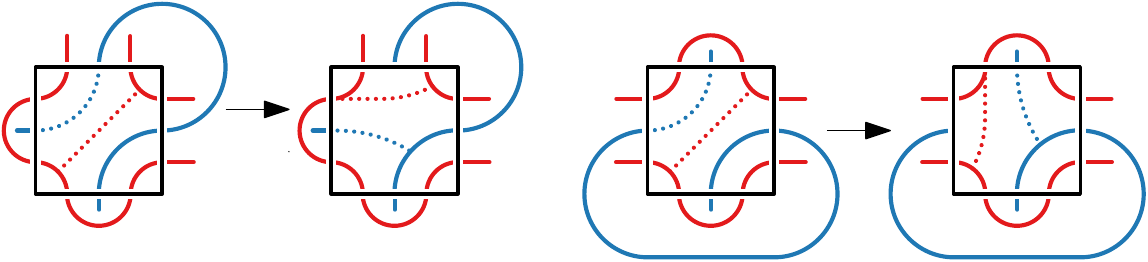}
\caption{There are two configurations for a 5-panel where both colors have at least two pieces. Both possible configuration can be simplified to a 4-panel.}
\label{fig:2-3-panel-not-needed}
\end{figure}

\subsection{Ensuring a 2-painting}
\label{ssec:simple}
We show that a fully 2-colored grid (rectangular and without white cells) even admits a $2$-painting, provided it admits any painting.
As an intermediate step, we first prove that a painting exists that uses only one blue piece in any panel.

\begin{lemma}\label{lem:semiregular}
If a fully 2-colored grid admits a painting, then it admits a painting in which each panel has at most one blue piece.
\end{lemma}
\begin{proof}[sketch]
    Since the grid admits a painting, we show how to modify the painting of each purple region $P$ to ensure that the lemma holds.
    We first create a blue spanning forest in the panels of $P$ that connects the panel-centers of adjacent panels.
    This ensures that each panel has exactly one blue piece inside, but may result in a disconnected blue polygon.
    However, since we know that a painting exists and the current solution maps to some forest in $G_b$, Lemma~\ref{lem:cotree} implies that its dual $G_r$ has a cycle around some tree in the forest.
    Hence, we can add connections between unconnected subpolygons to create a single blue polygon again, without disconnecting the red polygon.
	\hfill$\qed$
\end{proof}

    The above construction relies on the alternation of the blue and red intervals along the boundary of $P$.
    As there are no white cells we can guarantee this alternating pattern. 
    Indeed, the higher complexity with white cells is caused by long connections along a purple region's boundary that are needed to achieve this alternating pattern for a partially colored grid (e.g., Fig.~\ref{fig:complexPanels}).

	\begin{figure}[b]
		\centering
		\includegraphics[page=1]{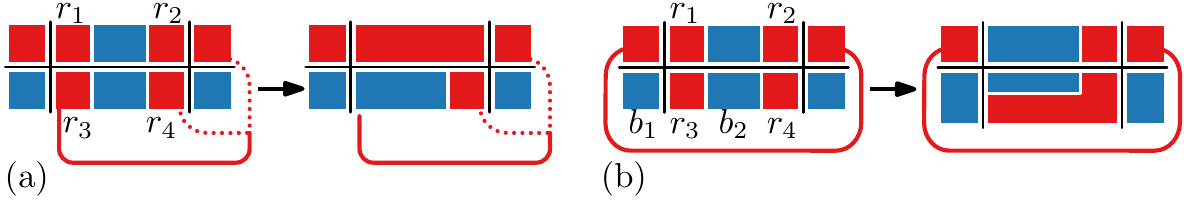}
		\caption{Reducing panel complexity when there are two red corners along the same panel side. (a) The corners are connected via adjacent (or the same) sides of the panel: connect $r_1$ and $r_2$, and recolor $r_3$ to blue. (b) The corners are connected via opposite sides: recolor $r_1$ to blue and connect $r_3$ and $r_4$ as well as $b_1$ and $b_2$.}
		\label{fig:simplepainting1}
	\end{figure}

\begin{theorem}\label{the:2painting}
If a fully $2$-colored grid admits a painting, then it admits a $2$-painting.
\end{theorem}

\begin{proof}[sketch]
	Since the fully 2-colored grid admits a painting, Lemma~\ref{lem:semiregular} implies that there is a painting $\Pi$ where the panel for every purple cell contains only a single blue piece.
For any purple cell with more than one red piece, we remove red pieces that only connect to one neighboring panel and recolor red corners to blue if the other three cells incident to that corner have a red corner as well. Now, the pattern of the panel matches one of the following four cases.

	\begin{enumerate}
		\item There are two red corners $r_1$ and $r_2$ on the same side of the panel. The connecting path exits the current panel via the same side and enters either on the same or adjacent side (see Fig.~\ref{fig:simplepainting1}(a)).
		\item There are two red corners $r_1$ and $r_2$ on the same side of the panel. The connecting path exits the panel via opposite sides of the panel (see Fig.~\ref{fig:simplepainting1}(b)). The blue piece connects only downwards in the panel below.
        \item There are two red corners $r_1$ and $r_2$ that do not share a common side of the panel. In this case the other corners are blue, otherwise one of the two previous cases applies (see Fig.~\ref{fig:simplepainting3}(a)). Furthermore, either $p_1$ or $p_2$ is blue.
        \item There are two red corners $r_1$ and $r_2$ that do not share a common side of the panel (see Fig.~\ref{fig:simplepainting3}(b)). Furthermore, both $p_1$ and $p_2$ are red.
	\end{enumerate}	
		
	\begin{figure}[t]
			\centering
			\includegraphics{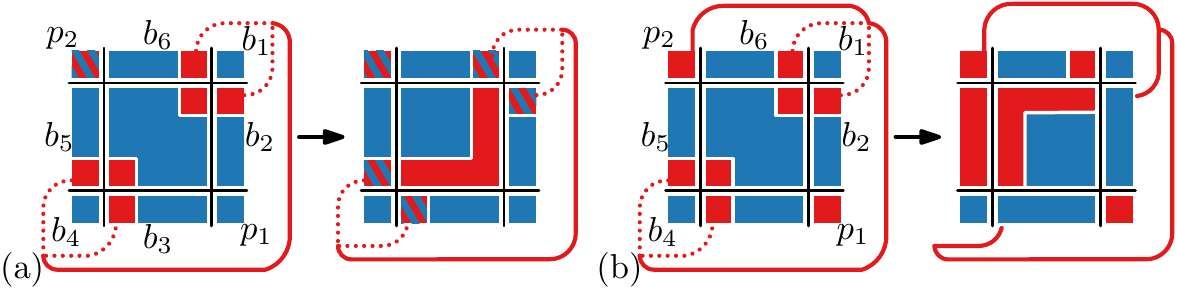}
		\caption{Two diagonally positioned red corners. The complexity of the panel can be reduced by introducing a red $L$-shape that connects all the red. (a) Reducing complexity if either $p_1$ or $p_2$ was blue. (b) Reducing complexity if both $p_1$ and $p_2$ were red.}
		\label{fig:simplepainting3}
	\end{figure}

We design a reduction rule for each case, as sketched in Fig.~\ref{fig:simplepainting1} and Fig.~\ref{fig:simplepainting3}. 
Repeated application of the reduction rules, interlaced with the reduction of the number of red pieces in a panel, results in a 2-painting.	
\hfill\qed
\end{proof}

\section{Conclusion}
\label{sec:conclusion}

We took the first steps towards investigating a disjoint-polygons representation for visualizing set memberships (hypergraphs).
We investigated the 2-color version in which each element is positioned as a cell in a (unit-)grid.
We showed how to test whether a disjoint-polygons representation is possible for a given 2-colored grid.
Moreover, we proved that if such a representation is possible, then we can also bound the complexity of the corresponding ``panels'' (the coloring of a single cell).
Each panel requires at most five colored pieces, and even only two pieces are sufficient when no white cells are present in the grid.

There are myriad options for further exploration.
As not all grids admit a painting, we could study minimizing the number of polygons of the same color.
We have not touched upon variants with more colors: does our approach readily generalize?
However, considering the restrictions already in the studied 2-color variant, it seems likely that many practical instances do not admit a painting.
If we allow rearranging elements, the 2-color variant becomes trivial, but is particularly interesting for multiple colors.
Finally, we may consider the situation where some cells have no assigned set of colors but may be painted using any subset of the colors.
Given enough such cells, the disjoint-polygons encoding can then represent more than Zykov-planar hypergraphs but cannot represent all planar supports.

\paragraph{Acknowledgments}
The authors would like to thank Jason Dykes for fruitful discussions at an early stage of this research.

\clearpage
\bibliographystyle{abbrv}
\bibliography{references}

\clearpage
\appendix

\section{Full Proofs}

\cloneclaim{Lemma \ref{the:dualPainting}}{A 2-colored grid $\Gamma$ in which each purple region $P$ has no holes and $\kappa(P)\leq 4$, admits a painting if and only if the corresponding graphs $G_r$ and $G_b$ are each other's exact duals.}
\begin{proof}
Suppose $\Gamma$ admits a painting $\Pi$. Then there exists a painting $\Pi'$ where, for any purple region $P$, all consecutive neighboring regions of the same color are connected through the boundary of $P$.
Let $n_r$ and $n_b$ denote the the number of vertices in $G_r$ (red regions) and in $G_b$ (blue regions) respectively.
We observe that the number of purple regions with $\kappa=4$, the number of edges in $G_r$, the number of edges in $G_b$, and the number of intersections is the same, say, $e$.
By construction, $n_b\geq f_r$ and $n_r\geq f_b$.
The purple regions with $\kappa=4$ can be painted to connect two adjacent red regions or two adjacent blue regions but never both.
As there are $n_r$ red regions, at least $n_r-1$ red edges are needed to connect all red vertices (regions). And similarly, at least $n_b-1$ edges are needed to connect all blue vertices (regions). Thus, if $\Gamma$ admits painting $\Pi'$ then $e\geq n_r+n_b-2$.

On the other hand, by Euler's formula, $n_r-e+f_r=2$ and $n_b-e+f_b=2$. Combining these equations and $n_b\geq f_r$, $n_r\geq f_b$, we derive that $e \leq n_r+n_b-2$. Thus, $e = n_r+n_b-2$, and $n_b=f_r$ and $n_r=f_b$. As each red edge intersects exactly one blue edge and vice versa, graphs $G_r$ and $G_b$ are each other's duals.

For the other direction, assume that $G_r$ and $G_b$ are dual graphs. By Lemma~\ref{lem:cotree} there exist two non-intersecting spanning trees of $G_r$ and $G_b$ that specify which adjacent regions of every purple region are to be connected to form connected red and blue polygons (see Fig.~\ref{fig:basecase-graphex}). Given the connectivity information for every purple region, it is straightforward to find the panels for all cells to obtain a connected painting: for every purple region, draw the two spanning trees, and connect any cell that does not have yet a blue or a red piece to the blue or red polygon respectively with a connection without introducing any crossings.
Thus, a 2-colored grid admits a painting if and only if $G_r$ and $G_b$ are dual graphs.
\hfill$\qed$
\end{proof}

\cloneclaim{Lemma \ref{lem:nonintersecting}}{Consider a painting $\Pi$ in which two blue and two red regions, adjacent to a purple region $P$, are connected through $P$. The corresponding vertices in the spiderweb gadget $W$ of $P$ can be connected by non-intersecting bridging paths.}
\begin{proof}
Let $u$ and $v$ be two vertices that represent two blue regions adjacent to $P$ and connected through $P$ in painting $\Pi$.
Let $u$ be on level $0$ of spoke $x$ and $v$ be on level $0$ of spoke $(x+2i)\mod 2k$ for some $1 \leq i \leq \lfloor k/2 \rfloor$.
We connect $u$ and $v$ by a bridging path as described above, which goes to level $2\lfloor (i+1)/2 \rfloor$.
Any connection of two red regions in $\Pi$ through the purple region $P$ cannot cross the connection between the regions that $u$ and $v$ represent in $\Pi$.
This means that the corresponding vertices must both lie on the same side of the bridging path between $u$ and $v$. Denote these vertices as $u'$ and $v'$, and let $u'$ have label $y$ and $v'$ have label $(y+2j)\mod 2k$ for some $1 \leq j\leq \lfloor k/2\rfloor$. 
The bridging path goes to level $2\lfloor j/2 \rfloor+1$.
There can be several cases:
\begin{itemize}
\item When $u'$ and $v'$ lie in between $u$ and $v$ moving in the clockwise order, we have that $j<i$. Thus, the two bridging paths cannot intersect, as the path from $u'$ to $v'$ goes to level $2\lfloor j/2 \rfloor+1 < 2\lfloor (i+1)/2 \rfloor$.
\item When $u'$ and $v'$ lie in between $v$ and $u$ moving in the clockwise order, we have that either $j>i$ or the two bridging paths lie on the opposite sides of the vertex on the innermost level. In the later case the two bridging paths cannot intersect as they are separated by the innermost level, and in the former case, the two bridging paths cannot intersect as the path from $u'$ to $v'$ goes to level $2\lfloor j/2 \rfloor+1 > 2\lfloor (i+1)/2 \rfloor$.\hfill\qed
\end{itemize}
\end{proof}

\cloneclaim{Lemma ~\ref{the:dualPainting2}}{A 2-colored grid $\Gamma$ in which each purple region has no holes admits a painting if and only if the corresponding $G_r$ and $G_b$ are each other's exact duals.}
\begin{proof}
The proof is similar to the proof of Lemma~\ref{the:dualPainting}.
Suppose $\Gamma$ admits a painting $\Pi$. Then there exists a painting $\Pi'$ where, for any purple region $P$, all consecutive neighboring regions of the same color are connected through the boundary of $P$. We can find non-intersecting trees in every spiderweb gadget $W$ corresponding to a purple region $P$ that connect all pairs of vertices of levels $0$ and $1$ in the same way as $\Pi$. First, we create the bridging paths from Lemma~\ref{lem:nonintersecting} for every pair of regions of the same color connected in $\Pi$ through $P$. The bridging paths do not create cycles in $S$, thus, every vertex in $W$ that is still not connected to the levels $0$ or $1$ can be connected to them by growing non-intersecting spanning forests from the vertices on the levels $0$ and $1$.

For the other direction, assume that $G_r$ and $G_b$ are dual graphs. By Lemma~\ref{lem:cotree} there exist two non-intersecting spanning trees of $G_r$ and $G_b$. These spanning trees provide the decisions of which adjacent regions of every purple region to connect, and the topology of the connections. Given the connectivity information, a painting can be constructed in a similar way to Lemma~\ref{the:dualPainting}. 
Thus, a 2-colored grid admits a painting if and only if $G_r$ and $G_b$ are dual graphs.
\hfill\qed
\end{proof}

\cloneclaim{Lemma \ref{lem:alternatingColors}}{If a structure $S$ exists with two adjacent cross-annulus connections $\gamma_x$ and $\gamma_y$ of the same color, possibly separated by non-crossing connections, then there also exists a structure $S'$ where $C_{S'}=C_S\setminus \{\gamma_y\}$.}
\begin{proof}
    \begin{figure}[h]
	\centering
		\includegraphics{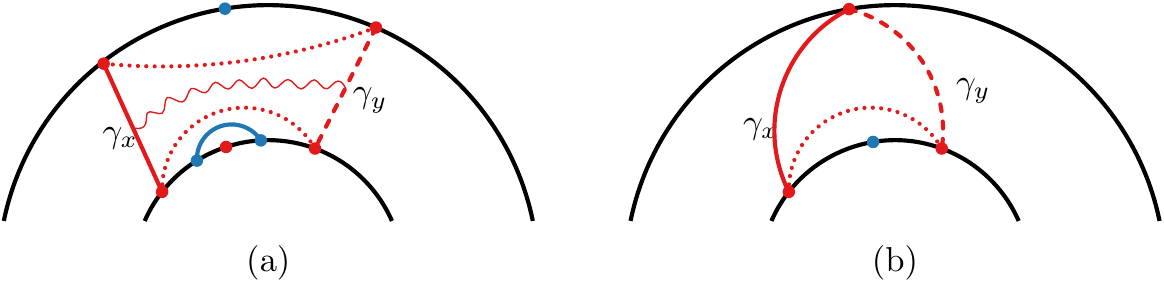}
		\caption{(a) When we have two adjacent cross-annulus connections of the same color, $\gamma_x$ and $\gamma_y$, we can remove $\gamma_y$ (dashed) and keep the same connectivity using two non-crossing connections (dotted). Note that, due to having no cross-annulus connections between $\gamma_x$ and $\gamma_y$, they are implicitly connected (zigzagged line). (b) The analogous case, if there are no blue regions between $\gamma_x$ and $\gamma_y$ on one of the sides.}
		\label{fig:annulusRemoveSameColorAdjacent}
	\end{figure}
    As there are no cross-annulus connections between $\gamma_x$ and $\gamma_y$, these cross-annulus connections must be connected inside the annulus.
    As a result, after removing $\gamma_y$ we can reconnect the structure by introducing two non-crossing connections between $\gamma_x$ and the two regions that $\gamma_y$ connects.
    This does not change the connectivity of the structure, but removes the explicit cross-annulus connection $\gamma_y$ (see Fig.~\ref{fig:annulusRemoveSameColorAdjacent}).
    Note that there cannot be a connection between $\gamma_x$ and $\gamma_y$ outside the annulus as this would disconnect any blue regions in between.
\end{proof}

\cloneclaim{Corollary \ref{cor:oneRedoneBlue}}{If a structure exists, then a structure also exists that has exactly one red and one blue connection across each annulus.}
\begin{proof}
By Lemma~\ref{lem:remove2} we can find a structure with at most three cross-annulus connections. By Lemma~\ref{lem:alternatingColors} a structure with three cross-annulus connections can be reduced to a structure with two cross-annulus connections.
\hfill\qed
\end{proof}

\cloneclaim{Lemma \ref{lem:threeBorders}}{If a 2-colored grid admits a painting, then it admits a painting where each panel $\pi$ has at most 3 intervals of alternating red and blue along each side.}

\begin{proof}
W.l.o.g. assume that a panel $\pi$ has at least 4 intervals of alternating red and blue on the left-side of $\pi$.
We consider the top-most four intervals that are inside $\pi$ and adjacent to the left-side of $\pi$ and w.l.o.g. we assume these intervals are ordered blue, red, blue, red.
As the painting is valid, both blue (/red) intervals are connected in the painting and since there can be no cycles there exists exactly one path connecting them.
For each interval we identify whether the path exiting or entering $\pi$ connects to the other interval of the same color (see Fig.~\ref{fig:borderColors}(a)).
It cannot be the case that the red and blue path both leave or exit $\pi$ in the same direction for the middle two intervals (see Fig.~\ref{fig:borderColors}(b)).
Assume w.l.o.g. that the middle blue connecting path exits $\pi$.
Any path connecting the blue intervals must separate the left side of the top red interval from the bottom red interval.
Hence the connecting path from the middle red interval cannot also exit $\pi$.

To reduce the number of intervals, we recolor the interval of the color whose connecting path exits $\pi$ (blue in Fig.~\ref{fig:borderColors}(a)) to the other color along the boundary of $\pi$.
To keep the blue polygon connected and remove the newly created red cycle, we move the other blue interval an epsilon distance inside $\pi$ and stretch it over the middle red interval. (see Fig.~\ref{fig:borderColors}(c)).
This reduces the number of intervals on the boundary of $\pi$ by two and can be repeated as necessary without affecting the validity of the solution.
\hfill\qed
\end{proof}

\cloneclaim{Lemma \ref{lem:semiregular}}{If a fully 2-colored grid admits a painting, then it admits a painting in which each panel has at most one blue piece.}
\begin{proof}
    Let $\Pi$ be a painting admitted by a fully 2-colored grid $\Gamma$.
    Any panel in $\Pi$ that is uni-colored trivially satisfies our lemma.
    Hence, we consider a purple region $P$ and show how to repanel it to ensure each panel has exactly one blue piece.
    Let $\Pi'$ be the repainted solution, which is identical to $\Pi$ except for the repainted cells in $P$.

    As there are no white cells, the blue and red intervals along the boundary of $P$ must alternate.
    Exceptions are the convex corners of $P$.
    At the convex corners there may be two blue cells adjacent to the same purple cell, separated by a red or purple cell (see Fig.~\ref{fig:toOneBlue}(a) top-left).
    As $\Gamma$ admits a painting these blue intervals must be connected through $P$.
    If both intervals are connected through the outside of $P$ then they isolate the red/purple corner cell from $P$.
	We refer to these connections through $P$ as the blue corners of $P$.
	At the outer boundary of the grid, the same can happen and we can treat them similarly. Note that, although the connection may now span multiple cells, it passes through each at most once.

    \begin{figure}
      \centering
      \includegraphics{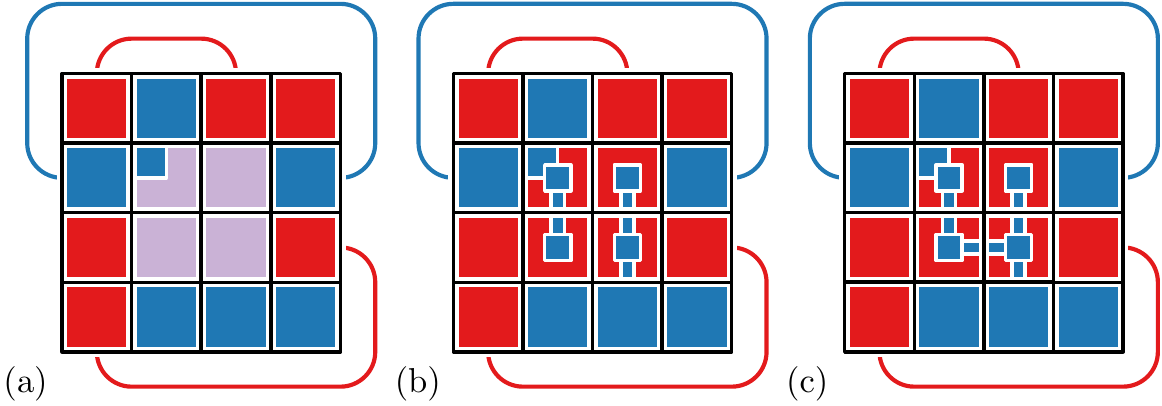}
      \caption{A purple region of four cells that has to be modified, including the direct neighborhood and topological connections in the original painting $\Pi$ outside $P$. (a) The top-left blue corner must be present. (b) Construction of the spanning forest. (c) Connecting two blue pieces breaks a cycle in the red polygon.}
      \label{fig:toOneBlue}
    \end{figure}

    We construct a new painting for $P$ as follows.
    First we place a small blue rectangle at the center of every panel in $P$.
    We then reconnect the blue corners of $P$ by filling the appropriate corners from this center rectangle.
    On these centers we build a spanning forest by connecting to the centers of adjacent cells.
    This spanning forest is built such that each tree in the forest is rooted in a distinct blue interval (see Fig.~\ref{fig:toOneBlue}(b)).
    By construction, all red pieces in $P$ forms a single polygon and thus red is connected.
    If there was only one blue polygon in $P$, then we are done as blue is connected and forms a single piece inside every cell in $P$.

    Assume there are two or more blue polygons in $P$.
    Since we know that a painting exists, only a red cycle can be broken when we connect two blue polygons.
    This is most easily seen in the constructed graphs $G_b$ and $G_r$ (see Theorem~\ref{the:dualPainting}): our constructed forest maps to a forest in $G_b$ and by Lemma~\ref{lem:cotree} we can arbitrarily complete this into a tree while ensuring that $G_r$ has a tree as well.
    Hence, we can pick two arbitrary polygons that have adjacent cells in $P$ and connect these (see Fig.~\ref{fig:toOneBlue}(c)).
    This reduces the number of blue polygons by one while keeping the red connected as the connection only cuts a red cycle.
    We repeat the above for every purple region, to ensure the eventual painting has a single blue piece in every panel.
	\hfill$\qed$
\end{proof}

\cloneclaim{Theorem \ref{the:2painting}}{If a fully $2$-colored grid admits a painting, then it admits a $2$-painting.}

\begin{proof}
	If there is a solution for a fully 2-colored grid $\Gamma$, then by Lemma~\ref{lem:semiregular} there is a solution $\Pi$ where every panel in a purple region only has a single blue piece.
That is, every panel only has red along the sides or corners.
All red pieces must include at least one corner by construction of Lemma~\ref{lem:semiregular}.
If a panel has a red piece that only connects to one neighboring panel then we retract it such that it is only adjacent to a single corner.
While a panel has more than one red piece we can fully remove any such red piece from the panel.

When four cells share a common point, each of them has a red corner at this common point, then we can color one of them blue, if it is the corner of a panel in a purple cell with another distinct red piece in it.
After repeated application of the above, any panel with multiple red components is in one of four cases:

	\begin{enumerate}
		\item There are two red corners $r_1$ and $r_2$ on the same side of the panel. The connecting path exists the current panel via the same side and enters either on the same or adjacent side. (see Fig.~\ref{fig:simplepainting1}(a)).
		\item There are two red corners $r_1$ and $r_2$ on the same side of the panel. The connecting path exits the panel via opposite sides of the panel (see Fig.~\ref{fig:simplepainting1}(b)). The blue piece connects only downwards in the panel below.
        \item There are two red corners $r_1$ and $r_2$ that do not share a common side of the panel. In this case the other corners are blue, otherwise one of the two previous cases applies (see Fig.~\ref{fig:simplepainting3}). Furthermore, either $p_1$ or $p_2$ is blue.
        \item There are two red corners $r_1$ and $r_2$ that do not share a common side of the panel. Furthermore, both $p_1$ and $p_2$ are red.
	\end{enumerate}
	
	We can reduce the complexity of each panel the following reduction rules:
	
	\begin{enumerate}
		\item Connect $r_1$ and $r_2$, and remove $r_3$ to break the red cycle (see Fig.~\ref{fig:simplepainting1}(a)).
		\item Connect $r_3$ and $r_4$, and connect $b_1$ and $b_2$ between $r_1$ and $r_3$.
                  We remove $r_1$ as it has become useless by connecting $b_1$ and $b_2$.
                  Similarly, we fully color the rest of the bottom panel red (see Fig.~\ref{fig:simplepainting1}(b)).
                  Note that the only option is for the blue piece in the bottom panel to connect solely to the panel below it. Assume that it also connected to the left panel. Then the left red piece must also connect to the left panel. When this red piece connects two different sides of this panel then it envelopes the blue corner and this panel cannot have anymore blue below. Contradiction. If the red piece ended in this panel, then it should include at least one corner, which it does not. Thus this situation can also not occur.
		\item If either $p_1$ or $p_2$ is blue, we connect $r_1$ and $r_2$ along this respective side of the panel (see Fig.~\ref{fig:simplepainting3}(a)). We connect either $b_1$ or $b_4$ to an adjacent blue polygon to break the red cycle and create a single blue polygon again.
		
\item When both $p_1$ and $p_2$ are red, we connect $r_1$ and $r_2$ along the side of the panel that is on the opposite side of the corner that the connection between both red pieces encloses (we pass $p_2$ in Fig.~\ref{fig:simplepainting3}(b)). This results in three blue polygons.
    Assume w.l.o.g. that the red corners are connected via a path leaving through the bottom and entering from the right.
    We join the blue polygons together by connecting $b_1$ and $b_2$, and $b_3$ and $b_4$, separating both sides of the connecting path.
We also recolor the border along the side of the adjacent panel that is not enclosed by the red connection ($b_5$ in this case) to ensure connectivity of the red.
	\end{enumerate}

Repeated application of the above reduction rules, interlaced with the reduction of the number of red pieces in a panel, must result in a 2-painting.
Every panel still has a single blue piece as no reduction rule increases the number of blue pieces in a panel.
Furthermore, as none of the above rules applies anymore, every panel has a single red polygon as well.
Thus, the painting is a 2-painting.
	\hfill$\qed$
\end{proof}

\end{document}